\documentclass[12pt]{article}
\usepackage{DEF}
\usepackage{DEF-script}

\def\lln{\ell_n}
\def\llnm{\ell_n^m}

\def\lan{\bar\ell_n}
\def\lanm{\bar\ell_n^m}

\def\rem{r^m}
\def\renm{r_n^m}
\def\Z{C}

\def\mX{\mathcal{X}}
\def\mN{\mathcal{N}}
\def\mY{\mathcal{Y}}

\def\tht{\th_\star}
\def\thn{\hat\th_n}
\def\thnm{\hat\th_n^m}

\def\grad{\nabla}
\def\hess{\nabla^2}

\def\app{_{\rm approx.}}

\def\topr{\to_{\rm p}}
\def\tod{\to_{\rm d}}

\title{Asymptotics of Monte Carlo maximum likelihood estimators
\footnote{Work partially supported by Polish National Science Center No. N N201 608 740}}

\author{
B{\l}a\.{z}ej Miasojedow
\footnote{Institute of Applied Mathematics and Mechanics,
University of Warsaw, Banacha 2, 02-097 Warsaw, Poland, {\tt  	B.Miasojedow@mimuw.edu.pl}
},
Wojciech Niemiro
\footnote{Faculty of Mathematics and Computer Science,
Nicolaus Copernicus University,  Chopina 12/18, 87-100, Toru\'n, Poland {\tt wniem@mat.umk.torun.pl}
and Institute of Applied Mathematics and Mechanics,
University of Warsaw, Banacha 2, 02-097 Warsaw, Poland, {\tt wniem@mimuw.edu.pl}
},\\
Jan Palczewski
\footnote{School of Mathematics, University of Leeds, Woodhouse Lane,  Leeds LS2 9JT , UK, {\tt J.Palczewski@leeds.ac.uk}
}
and Wojciech Rejchel
\footnote{corresponding author, Faculty of Mathematics and Computer Science,
Nicolaus Copernicus University,  Chopina 12/18, 87-100, Toru\'n, Poland, +48 566112943,
{\tt wrejchel@gmail.com}}
}

\date{}
\begin{document}

\maketitle

\begin{abstract}
We describe Monte Carlo approximation to the maximum likelihood estimator in  models
with intractable norming constants and explanatory variables.
We consider both sources of randomness (due to the initial sample and to Monte Carlo simulations) and prove asymptotical normality of  the estimator.

\end{abstract}
\bigskip\goodbreak


\bigskip\goodbreak
\section{Introduction}\label{Sec:Introduction}

Maximum likelihood (ML) is a well-known and often used method in estimation
of parameters in statistical models. However, for many complex models exact calculation of such estimators is very difficult or
impossible. Such problems arise if considered densities are known only up to intractable
norming constants, for instance in Markov random fields or spatial statistics. The wide range of applications
of models with unknown norming constants is discussed e.g.\  in \cite{MPRB2006}.
Methods proposed to overcome the problems with computing ML estimates in such models include, among others,
maximum pseudolikelihood (MPL) \cite{Besag1974} or
Monte Carlo maximum likelihood (MCML) \cite{Cappe2002, GeyerThom1992, Geyer1994, GeyerSung2007}.
MPL estimators are easy to compute but not efficient. This
is demonstrated e.g.\ in \cite{WuHu1997} for an important autologistic spatial model via a simulation study.
Comparison of MLP or ,,coding method'' with MCML is also discussed in \cite{HuWu1998}.
In our paper we focus on MCML.

In influential papers \cite{GeyerThom1992, Geyer1994} the authors prove consistency and asymptotic normality of
MCML estimators under the assumption that the initial sample is fixed, and only the Monte Carlo sample size tends to infinity.
Both sources of randomness (one due to the initial sample and the other due to Monte Carlo simulations) are considered in
\cite{Cappe2002, GeyerSung2007, Imput2010}. Authors of the first mentioned paper apply the general importance sampling recipe. They show that for their scheme of simulations,
the Monte Carlo sample size has to grow exponentially fast to ensure consistency of the estimator.
As the remedy for this problem they propose to use a preliminary estimator which is consistent. Another
possibility to overcome this problem is proposed in \cite{GeyerSung2007}.
The log-likelihood is first decomposed into independent summands and then importance sampling is applied.
Papers \cite{Cappe2002, GeyerSung2007}   describe asymptotic properties of MCML estimators for models with missing data.
In our paper we consider models with intractable norming constants and explanatory variables. We apply argumentation similar to
\cite{GeyerSung2007} in our setting.

We consider a parametric model with covariates
\begin{equation}\nonumber
p(y|x,\th)=\frac{1}{\Z(x,\th)}f(y|x,\th),
\end{equation}
where $y \in \mY \subset \mathbb{R}^d$ is a response variable, $x \in \mX \subset
\mathbb{R}^l$ is a covariate or ``explanatory'' variable (random or deterministic), $\th \in \mathbb{R} ^p$ is a parameter
describing the relation between $y$ and $x$. The norming constant,
\begin{equation}\nonumber
\Z(x,\th) = \int f(y|x,\th) dy,
\end{equation}
is difficult or intractable.

Assume that the data consist of $n$ independent observations $(Y_1,X_1),\ldots,$ $(Y_n,X_n).$
If we regard covariates as random, then we assume that these pairs form an i.i.d.\  sample from a joint distribution
with a density $g(y,x)$. Alternatively, $x_i$ can be regarded as deterministic and then we assume that random variable
$Y_i$ has a probability distribution $g_i$ which depends on $x_i$. Both cases can be analysed very similarly.
For simplicity we focus attention on the model with random covariates.
It is not  necessary to assume that $g(y | x)=p(y|x,\th_0)$ for some $\th_0$. The case when no such $\th_0$ exists, i.e.\
the model is misspecified, makes the considerations only slightly more difficult. Thus,
let us consider the following log-likelihood
\begin{eqnarray}\nonumber \label{eq:likn}
 \lln(\th) &=&\log p(Y_1,\ldots,Y_n|X_1,\ldots,X_n,\th)\\
           &=&\sum_{i=1}^{n} \log  f(Y_i|X_i,\th)- \sum_{i=1}^{n} \log {\Z(X_i,\th)}  .
\end{eqnarray}
The first term in (\ref{eq:likn}) is easy to compute while the second
one is approximated by Monte Carlo (MC). Let $h(y)$ be an importance sampling (instrumental)
distribution and  note that
\begin{equation}\nonumber
{\Z(x,\th)}=\int f(y|x,\th)\d y=\int \frac{f(y|x,\th)}{h(y)}h(y)\d y=\Ex_{Y\sim h}\frac{f(Y|x,\th)}{h(Y)}.
\end{equation}
Thus, an MC approximation of the log-likelihood $\lln(\th)$ is
\begin{equation} \label{eq:fr}
\llnm(\th) =\sum_{i=1}^{n }\log f(Y_i|X_i, \th)-\sum_{i=1}^{n}\log C_m(X_i, \th) ,
\end{equation}
where
$$C_m(x,\th) = \frac{1}{m}\sum_{k=1}^{m}\frac{f(Y^k|x, \th)}{h(Y^k)}$$
and $Y^1,\ldots, Y^m$ is a  sample drawn form $h$.

Let us note that the general Monte Carlo recipe 
can also lead to approximation schemes different from \eqref{eq:fr}. For instance, we could generate $n$ independent
MC samples instead of one, i.e.\  $Y_i^1,\ldots,Y_i^m \sim h_i, i=1,\ldots,n$ and use
$i$th sample to approximate $\Z (x_i, \theta).$ Using this scenario one can obtain
estimators with better convergence rates, but at the cost of increased  computational
complexity. Another scheme, proposed in \cite{Cappe2002}, approximates the log-likelihood by
\begin{equation}\label{Cappe}
\sum_{i=1}^{n }\log f(Y_i|X_i, \th)-\log \frac{1}{m}\sum_{k=1}^{m}\prod_{i=1}^{n}\frac{f(Y_i^k|X_i, \th)}{h_i(Y_i^k).
}\end{equation}
{However, this scheme leads to estimators with unsatisfactory asymptotics unless a preliminary estimator is used.
Thus, we focus our attention only on $(\ref{eq:fr}).$

Let $\thn$ be a maximizer of $\lln(\th)$ (a genuine maximum likelihood estimator).
It is well-known that under some regularity assumptions \cite{Pollard1984, vaart1998}
\begin{equation}\nonumber 
\thn \sim\app \mathcal{N} \left(\tht,\frac{1}{n}D^{-1}VD^{-1}\right),
\end{equation}
where $\tht$ is a  maximizer of $\Ex_{(Y,X)\sim g} \log p(Y|X,\th)$,
i.e.\ Kullback-Leibler projection, $D=\Ex_{(Y,X)\sim g} \hess  \log p(Y|X,\tht)$ and $V=\VAR_{(Y,X)\sim g}\grad \log p(Y|X,\tht)$.
Symbols $\grad$ and $\hess$ denote derivatives with respect to $\th$ and $\VAR$ stands for the variance-covariance matrix.
In Theorem \ref{with_cov} we will prove that the  maximizer of (\ref{eq:fr}), denoted by  $\thnm$, satisfies
\begin{equation}\label{asymp}
\thnm
             \sim\app \mathcal{N} \left(\tht,D^{-1}\left(\frac{V}{n}+\frac{W}{m}\right)D^{-1}\right),
\end{equation}
where the matrix $W$ will be given later. Formula (\ref{asymp}) means that the estimator
$\thnm$ behaves like a normal vector with the mean $\tht$ when both the initial sample size $n$ and the
Monte Carlo sample size $m$ are large. Note that the first component of
the asymptotic variance in (\ref{asymp}),  $D^{-1}VD^{-1}/n$,  is the same as the asymptotic variance
of the maximum likelihood  estimator $\thn.$ The second component, $D^{-1}WD^{-1}/m$, is due to Monte Carlo randomness.
Furthermore, if $m$ is large, then asymptotic behaviour of $\thnm$ and $\thn$ is similar.
If the model is correctly specified, that is $g(y|x) = p(y|x , \th_0)$ for some $\th_0,$
then $\tht = \th_0$ and $D=-V$ (under standard assumptions on passing the derivative
under the integral sign).

The choice of the instrumental distribution $h$ affects $W$ and thus the asymptotic efficiency of MCML.
In \cite[Equation (2.17)]{MNPR2014b} a formula for optimal $h$ is derived (this $h$ minimizes the trace of $W$ in a model without covariates). 
This result may be of some theoretical interest but has a limited practical value, because the optimal $h$ can be very difficult to sample from.
On the other hand, a more practical approach, suggested by several authors, e.g.\  \cite{Cappe2002, Imput2010}, is
to select some distribution in the underlying parametric family, i.e.\ to put
\begin{equation}\nonumber
h(y)=p(y|\psi)=\frac{1}{\Z(\psi)}f(y|\psi), 
\end{equation}
for some fixed $\psi\in\Rl^p$ (here we restrict attention to models without covariates). It is natural to guess that a ``good choice'' of $\psi$ should be close
to the target, $\tht$. Since $\tht$ is unknown, one can use a preliminary estimator. Such a choice of $h$ 
is recommended in \cite{Cappe2002, Imput2010}. In the first of the cited papers, theoretical results are given
which justify using a consistent preliminary estimate of $\tht$ as $\psi$, compare \cite[Theorems 4 and 7]{Cappe2002}. 
However, the results are about sampling scheme \eqref{Cappe}. In \cite{Imput2010}, sampling scheme $(\ref{eq:fr})$ is considered and the
choice of $\psi$ near $\tht$ is recommended on heuristical grounds. In fact the intuition behind this choice turns out to be wrong, as
demonstrated by the following toy example.

\def\yn{\bar Y_n}
\def\ym{\bar Y^m}

\begin{exam} Let $\Y=\{0,1\}$ and $f(y|\th)=\e^{\th y}$ for  $\th\in\Rl$. Of course, the norming constant
$C(\th)=1+\e^\th$ is easy and there is no need to apply MCML, but the simplicity of this model will allow us 
to clearly illustrate our point. Assume we have an i.i.d.\ sample $Y_1,\ldots,Y_n$ from $f(\cdot|\tht)/C(\tht)$.
The MLE is $\thn=\log(\yn/(1-\yn))$, where $\yn=n^{-1}\sum_{i=1}^n Y_i$. Now suppose that we use MCML approximation 
\eqref{eq:fr} with $h(y)=f(y|\psi)/{\Z(\psi)}$. It can be easily shown that the asymptotic variance $W$ (now a scalar) 
is minimum for $\psi_\star=0$ -- and not for $\psi=\tht$! The following direct derivation explains this fact.
The formula \eqref{eq:fr} now assumes the form 
\begin{equation}\nonumber
 \llnm(\th)=n\th\yn-n\log\left(\frac{1}{m}\sum_{k=1}^m\e^{(\th-\psi)Y^k}\right)-n\log\Z(\psi).
\end{equation}
On noting that 
\begin{equation}\nonumber
 \frac{1}{m}\sum_{k=1}^mY^k\e^{(\th-\psi)Y^k}=\ym\e^{\th-\psi},\quad 
 \frac{1}{m}\sum_{k=1}^m\e^{(\th-\psi)Y^k}=\ym\e^{\th-\psi}+(1-\ym)
\end{equation}
we see that the equation $\grad \llnm(\th)=0$ is equivalent to 
 \begin{equation}\nonumber
  \yn-\dfrac{\ym\e^{\th-\psi}}{\ym\e^{\th-\psi}+(1-\ym)}=0.
 \end{equation}
 After elementary computations we obtain that the solution $\thnm$ of this equation is
 \begin{equation}\nonumber
  \thnm=\log\dfrac{\yn}{1-\yn}+\psi-\log\dfrac{\ym}{1-\ym}.
 \end{equation}
Let us rewrite this expression as follows:
\begin{equation}\nonumber
  \thnm=\thn+\psi-\hat\psi^m,
 \end{equation}
where $\hat\psi^m$ is an ML estimate of $\psi$ based on the MC sample. It is clear that 
$\sqrt{m} (\psi-\hat\psi^m)\to\N(0,\e^{-\psi}(1+ \e^\psi)^2)$, independently of $\th$. The
asymptotic variance of the MC error is minimum for $\psi_\star=0$. The overall error of MCML 
is the sum of two independent terms $(\thn-\tht)+(\hat\psi^m-\psi_\star)$.
\end{exam}

Asymptotic properties of MCML estimator (consistency, rates of convergence, asymptotic normality) can be obtained using
standard statistical methods from the empirical processes theory \cite{Pollard1984, vaart1998}. However, these tools should be
adjusted to the model with double randomness when both sample sizes $n$ and $m$ tend to infinity simultaneously.
This adaptation makes our proofs very arduous and technical despite the fact that the main ideas are rather clear.
Therefore, to make the paper more transparent we present only the proof of asymptotic normality. This result
is the most important from a practical point of view. Moreover, the argumentation used in proving this property
illustrates well how to adapt standard methods to the double randomness setup. Similar
adaptation can be used to obtain consistency and the rate of convergence of the MCML estimator.
Since the proof of (\ref{asymp}) for the model with covariates is rather complicated, we begin in Section 2 with a model without covariates. It is extended to the general case in Section 3.

As we have already mentioned, related results on MCML for missing data models can be found in \cite{Cappe2002, GeyerSung2007}.
In particular, our theorems are of similar form as those in \cite{GeyerSung2007}. However,
models with intractable norming constants and observable covariates,  considered in our paper, are
more difficult to analyse. 
Let us also mention that for the missing data models there exists another powerful tool for computing
maximum likelihood estimates, namely the EM (expectation-maximization) algorithm \cite{DempsterLairdRubin1977}.
The expectation step (E-step) can be implemented using MC computations resulting in Monte Carlo EM (MCEM) algorithm which
has been examined in several papers \cite{WeiTanner1990, LevineCasella2001, FortMoulines2003}. MCEM cannot be applied to models with intractable norming constants and observable covariates. This points to particular importance of MCML in this setting and motivates examination of its behaviour.

\section{Model without covariates}

First, we consider a model without covariates
\begin{equation}\nonumber
p(y|\th)=\frac{1}{\Z(\th)}f(y|\th)
\end{equation}
with an intractable norming constant $\Z(\th ) = \int f(y|\th) dy.$
Assume we have an i.i.d.\  sample $Y_1,\ldots,Y_n\sim g(y)$.
Similarly to the general case, we allow for misspecification of the model, i.e.\
we do not assume $g(y)=p(y|\th_0)$ for some $\th_0$. In what follows, $\tht$ is a maximizer of $\Ex_{Y\sim g} \log p(Y|\th)$, i.e.\
the Kullback-Leibler projection. The MC approximation (\ref{eq:fr}) multiplied by $\frac{1}{n}$ is denoted by
\begin{equation}\nonumber 
\begin{split}
\lanm(\th) &=\frac{1}{n}\sum_{i=1}^{n }\log f(Y_i|\th)-\log \frac{1}{m}\sum_{k=1}^{m}\frac{f(Y^k|\th)}{h(Y^k)}
           =\lan(\th) - \rem(\th),
\end{split}
\end{equation}
where
\begin{eqnarray*}
\lan(\th)  &=& \frac{1}{n} \sum_{i=1}^{n}\left[ \log  f(Y_i|\th)-\log {\Z(\th)}\right],  \\
\rem(\th) &=&\log \frac{1}{m}\sum_{k=1}^{m}\frac{f(Y^k|\th)}{h(Y^k)}-\log {\Z(\th).}
\end{eqnarray*}
Now we can state the main result of this section.

\begin{thm}\label{without_cov}
For some $\delta >0$ let $U= \{\th : |\th - \tht| \leq \delta\} $ be a
neighbourhood of $\tht.$ If the following assumptions are satisfied:
\begin{enumerate}
\item second partial derivatives of $f(y |\th)$ with respect to $\th$ exist and are continuous for all $y$, and can be passed under the integral sign in
$\int f(y| \th) dy,$
\item $\sqrt{\min(n,m)} \left( \thnm - \tht \right) = O_p (1),$
\item matrices
\begin{eqnarray*}
V&=&\VAR_{Y\sim g}\grad \log p(Y|\tht), \\
D&=&
\Ex_{Y\sim g} \hess  \log p(Y|\tht)\\
W &=& \frac{1}{\Z^2(\tht)} \VAR_{Y \sim h}
\left[\dfrac{\grad f(Y|\tht)}{h(Y)}-\dfrac{\grad \Z(\tht)}{\Z(\tht)}\dfrac{ f(Y|\tht)}{h(Y)}\right]
\end{eqnarray*}
exist and $D$ is negative definite,
\item function $D(\th) =  \Ex_{Y\sim g} \hess \log p(Y|\th) $ is continuous at $\tht$,
\item
$
\sup\limits_{\th \in U} | \hess \lan (\th) - D(\th)| \topr 0 , \quad n
\rightarrow \infty,
$
\item
$
\sup\limits_{\th \in U}  | \hess C_m (\th) - \hess C(\th)| \topr 0, \quad
m \rightarrow \infty,
$
\end{enumerate}
then
\begin{equation}\nonumber
\left( \frac{V}{n} + \frac{W}{m} \right)^{- \frac{1}{2}}\: D \left(\thnm -
\tht
\right) \tod \mathcal{N} (0,I), \quad n,m \rightarrow \infty.
 \end{equation}
\end{thm}

Note that 1 and 3 are rather standard regularity assumptions. Condition 2
stipulates the square root consistency of the MCML estimator. If the MC approximation 
$\lanm(\th)$ is concave (as in the example studied below), then assumption 2 is automatically fulfilled \cite{Niemiro92}. Otherwise, it can be deduced from more explicit
assumptions by adapting standard
methods from the empirical processes theory \cite{Pollard1984, vaart1998} to the
double randomness problem. For simplicity, we do not explore this topic. We just choose condition 2 as a starting point of our
argumentation (which is by itself quite complicated).

We shall show that conditions 4 - 6 are satisfied for exponential families, i.e.\ if
$$
f(y|\th)= \exp( \th\t W(y))
$$
with $W(y) = \left( W_1(y), \ldots, W_p(y)\right).$ We can easily verify
that $\hess \log p(y|\th)=  - \hess \log C(\th), $ so assumptions 4 and 5 are
obviously fulfilled. Thus, condition 6 is the last one to establish.
Function  $ \hess
C_m(\th)$ is  matrix-valued, so it is enough to  prove that for each component (that is for each $r,s=1, \ldots, p$)
\begin{equation} \label{VC}
\sup\limits_{\th \in U} \left| \left[ \hess C_m(\th) \right]_{rs} - \left[ \hess C(\th)  \right]_{rs} \right| \to_p 0, \quad m \rightarrow \infty.
\end{equation}
Consider a family of functions
\begin{equation} \label{VCfam}
\left\{ \left[ \frac{\hess f(y | \th) }{h(y)} \right]_{rs } =
\exp( \th\t W(y)) \, \frac{W_r(y) W_s(y)}{h(y) } : \th \in U
\right\}.
\end{equation}
The set $U$ is compact, so to obtain \eqref{VC} it is sufficient to
assume functions} in \eqref{VCfam}
are dominated by an integrable function (see \cite[Theorem 16(a)]{Ferguson1996},
\cite[Example 19.8]{vaart1998}), i.e. for each $r,s$ there is a
function $\eta$ such that $\Ex_{Y\sim h} \eta(Y) < \infty$ and $\big|
[ \hess f(y | \th) / h(y) ]_{rs } \big| \le \eta(y)$ for each $\th, y$.

\begin{proof}[Proof of Theorem \ref{without_cov}]
Without loss of generality we can assume that $\tht = 0.$ 
First we assume that $\frac{n}{n+m}\to a$ and consider three cases corresponding to rates at which $n$ and $m$ go to infinity:
$0<a<1$,  $a=0$ and $a=1$. Once our theorem is proved in these three special cases, standard application
of the subsequence principle shows that it is valid in general (for $n\to \infty$ and $m\to \infty$ at arbirary rates).

We begin with the case  $0<a<1$. It is well-known (see
\cite[Theorem VII.5]{Pollard1984})
that we need to prove
\begin{equation}\label{as_norm1}
\left( \frac{V}{n} + \frac{W}{m} \right)^{- \frac{1}{2}} \grad \lanm(0)
 \rightarrow_d \mathcal{N} (0,I), \quad n,m \rightarrow \infty
\end{equation}
and for every $M>0$
\begin{equation}\label{sup1}
(n+m)  \sup_{|\th|\leq \frac{M}{\sqrt{n+m}}} \left|\lanm(\th)-\lanm(0)- \th\t \grad\lanm(0)-\frac{1}{2}\th\t D\th\right| \topr 0, \quad n,m \rightarrow \infty.
\end{equation}

To obtain (\ref{as_norm1})  notice that
\begin{equation}\label{as_norm11}
\sqrt{n+m} \grad \lanm(0) = \sqrt{\frac{n+m}{n}} \: \sqrt{n} \grad \lan(0)
- \sqrt{\frac{n+m}{m}} \: \sqrt{m} \grad \rem(0)
\end{equation}
and the terms on the right hand side  in (\ref{as_norm11}) are independent.
We can calculate the gradient
\begin{equation}\nonumber
\grad\rem(0)  =\frac{\dfrac{1}{m}\sum\limits_{k=1}^{m}\left[\dfrac{\grad f(Y^k|0)}{h(Y^k)}-
                        \dfrac{\grad\Z(0)}{\Z(0)} \dfrac{f(Y^k|0)}{h(Y^k)}\right]}
                 {\dfrac{1}{m}\sum\limits_{k=1}^{m} \dfrac{f(Y^k|0)}{h(Y^k)}}.
\end{equation}
Therefore, by LLN, CLT and Slutsky's theorem
we have that $\sqrt{m} \grad \rem(0) \tod \mN (0,W) $ and
$\sqrt{n} \grad \lan(0) \tod \mN(0,V) $ which implies
\begin{equation}\nonumber
\sqrt{n+m} \grad \lanm(0) \tod \mN \left(0, V/a + W/ (1-a)\right), \quad n,m
\rightarrow \infty.
\end{equation}
Thus, we obtain (\ref{as_norm1}) since
$$\sqrt{n+m} \left(V/a + W/ (1-a)\right)^{-\frac{1}{2}}
\left(V/n + W/ m \right)^{\frac{1}{2}} \rightarrow I \quad n,m \rightarrow \infty.$$

Now we focus on (\ref{sup1}). Using Taylor expansion it can be bounded by
\begin{equation}\label{sup11}
\frac{M^2}{2}\!\! \left( \sup_{\th \in U^m_n} \left|\hess \lan(\th)-D(\th)\right| + \sup_{\th \in U_n^m}\left|D(\th)-D(0)\right|  +
 \sup_{\th \in U^m_n} \left|\hess \rem(\th) \right| \right)
\end{equation}
for $U^m_n = \{ \th : |\th| \leq \frac{M}{\sqrt{n+m}} \}.$
First two terms in (\ref{sup11}) tend to zero in probability by
assumptions 4 and 5. We prove that assumption 6 implies convergence
to zero in probability of the third term
 in (\ref{sup11}). Calculating  the second
derivative we get
\begin{equation}\nonumber
\hess \rem(\th)=\frac{  \hess C_m(\th)}{ C_m(\th)}-\frac{\grad C_m(\th)  \grad\t C_m(\th)}{C_m^2(\th)}
-\frac{\hess \Z(\th)}{\Z(\th)}+\frac{\grad \Z(\th)\grad\t \Z(\th)}{\Z^2(\th)}.
\end{equation}
Therefore
\begin{eqnarray*}
&&\sup_{\th \in U} \left|\hess \rem(\th) \right| \leq \sup_{\th \in U}
\frac{  |\hess C_m(\th)| \, |C_m(\th) - C(\th)|}{ C_m(\th) C(\th)} +
\sup_{\th \in U} \frac{  | \hess C_m(\th) - \hess C(\th)|}{ C(\th)}\\
&&+ \sup_{\th \in U}
\frac{  |\grad C_m(\th)|^2 \, |C_m^2(\th) - C^2(\th)|}{ C_m^2(\th) C^2(\th)}+
\sup_{\th \in U} \frac{  | \grad C_m(\th)  \grad\t C_m(\th)-
\grad C(\th)\grad\t C(\th)|}{ C^2(\th)}.
\end{eqnarray*}
Note that continuous functions $C(\th), |\grad C (\th)|, |\hess C (\th)|$ are bounded on the compact set $U$, in
particular function $C(\th)$ is separated from zero. Therefore, all we need is
assumption 6 and
\begin{eqnarray}\label{sup12}
\sup_{\th \in U}  |  C_m (\th) -  C(\th)| \topr 0, \quad
m \rightarrow \infty, \\ \label{sup125}
\sup_{\th \in U}  | \grad C_m (\th) -  \grad C(\th)| \topr 0, \quad
m \rightarrow \infty.
\end{eqnarray}
However, uniform convergence in (\ref{sup12}) and (\ref{sup125}) easily follows from Taylor expansion, LLN and assumption 6. For instance, for some $\th' \in (0,\th)$
\begin{equation} \nonumber
\grad C_m(\th)-\grad C( \th) =\grad C_m(0) - \grad C( 0) +
\left[ \hess C_m (\th') - \hess C(\th')\right] \th,
\end{equation}
so
\begin{equation} \nonumber
\sup_{\th \in U } |\grad C_m(\th)-\grad C( \th) | \leq|\grad C_m(0) - \grad C( 0)| +
 \delta \, \sup_{\th \in U } \left| \hess C_m (\th) - \hess C(\th)\right|.
\end{equation}
Thus, the proof in the case $0<a<1$ is finished. For $a=0$ or $a=1$ we
proceed similarly. For example, if $a=0$, then we should prove
an analog of (\ref{sup1}), namely for every $M>0$
\begin{equation}
\label{a0}
n  \sup_{|\th|\leq \frac{M}{\sqrt{n}}} \left|\lanm(\th)-\lanm(0)- \th\t \grad\lanm(0)-\frac{1}{2}\th\t D\th\right| \topr 0, \quad
n,m \rightarrow \infty.
\end{equation}
Argumentation is almost the same as in the proof of (\ref{sup1}). To obtain
(\ref{as_norm1})  in this case note that
\begin{equation}\label{as_norma0}
\sqrt{n} \grad \lanm(0) =  \sqrt{n} \grad \lan(0)
- \sqrt{\frac{n}{m}} \: \sqrt{m} \grad \rem(0).
\end{equation}
Therefore, expression (\ref{as_norma0}) tends in distribution to $\mathcal{N} (0, V) .$
Moreover,
$$\sqrt{n} V^{-\frac{1}{2}}
\left(V/n + W/ m \right)^{\frac{1}{2}} \rightarrow I, \quad n,m \rightarrow \infty.$$

\end{proof}

\section{Model with covariates}

Let us return to the general case and state the main theorem of the paper. We need new notation:
\begin{eqnarray*}
\phi(y|x) &=&\left[\dfrac{\grad f(y|x,\tht)}{h(y)}-\dfrac{\grad \Z(x,\tht)}{\Z(x,\tht)}\dfrac{ f(y|x,\tht)}{h(y)}\right] \frac{1}{\Z(x,\tht)} \;,\\
\rem_n(\th) &=& \frac{1}{n} \sum_{i=1}^n \left[ \log \frac{1}{m} \sum_{k=1}^m
\frac{f(Y^k|X_i,\th)}{ h(Y^k)} - \log \Z (X_i, \th)
\right]\;.
\end{eqnarray*}
Then $\lanm (\th) = \lan(\th) - \rem_n (\th).$

\begin{thm}\label{with_cov}
For some $\delta >0$ let $U= \{\th : |\th - \tht| \leq \delta\} $ be a
neighbourhood of $\tht.$ Suppose the following assumptions are satisfied:
\begin{enumerate}
\item second partial derivatives of $f(y |x,\th)$ with respect to $\th$ exist and are continuous for all $y$ and $x$, and may be passed under the integral sign in
$\int   f(y|x, \th) dy$ for fixed $x$,
\item $\sqrt{\min(n,m)} \left( \thnm - \tht \right) = O_p (1),$
\item matrices
\begin{eqnarray*}
V&=&\VAR_{(Y,X)\sim g}\grad \log p(Y|X,\tht), \\
D&=&
\Ex_{(Y,X)\sim g} \hess  \log p(Y|X,\tht)
\end{eqnarray*}
and the expectation $\tilde{W} = \Ex_{Y \sim h, X\sim g}
|\phi(Y|X)|^2$
exist and $D$ is negative definite,
\item function $D(\th) = \Ex_{(Y,X)\sim g} \hess  \log p(Y|X,\theta) $ is continuous at $\tht$,
\item $\sup_{\th \in U} | \hess \lan (\th) - D(\th)| \rightarrow_P 0 , \quad n
\rightarrow \infty,$
\item
\begin{enumerate}
 \item  $\sup_{x \in \mX} |C_m (x,\tht) - C(x,\tht)|  \topr 0, \quad
m \rightarrow \infty,$
 \item  $\sup_{x \in \mX} |\nabla C_m (x,\tht) - \nabla C(x,\tht)| \topr 0, \quad
m \rightarrow \infty,$
\item  $\sup\limits_{\th \in U ,x \in \mX} |\hess C_m (x,\th) - \hess C(x,\th)| \topr 0, \quad
m \rightarrow \infty,$
\end{enumerate}
\item there exist constants $\alpha > 0,$ $K > 0$ such that for each $x \in \mX$ and $\th \in U$
$$
 \alpha \leq C(x,\th) \leq K, \quad
| \grad C(x,\th)| \leq K, \quad
| \hess C(x,\th)| \leq K.
$$
\end{enumerate}
Then matrix
\begin{equation}\nonumber
W=\VAR_{Y \sim h} \: \Ex_{X\sim g}
\; \phi(Y|X)
\end{equation}
is finite and
\begin{equation}\nonumber
\left( \frac{V}{n} + \frac{W}{m} \right)^{- \frac{1}{2}}\: D \left(\thnm -
\tht
\right) \tod \mathcal{N} (0,I), \quad n,m \rightarrow \infty.
 \end{equation}
\end{thm}

We discuss assumptions in Theorem \ref{with_cov} for functions $f(y|x,\th)$ belonging 
to the exponential family at the end of this section.

\begin{proof}
Without loss of generality we can assume that $\tht = 0.$

Similarly to the proof of Theorem \ref{without_cov} we consider three cases: $0<a<1$,  $a=0$ and $a=1$, where  $\frac{n}{n+m}\to a$.
Finally, we complete the proof by using the subsequence principle.

We focus on the case $0<a<1$, because for  $a=0$ or $a=1$ we proceed in a similar
way (cf.\ the proof of Theorem \ref{without_cov}).  It is well-known (see
\cite[Theorem VII.5]{Pollard1984})
that we need to prove that for every $M>0$
\begin{equation}\label{sup2}
(n+m)  \sup_{|\th|\leq \frac{M}{\sqrt{n+m}}} \left|\lanm(\th)-\lanm(0)- \th\t \grad\lanm(0)-\frac{1}{2}\th\t D\th\right| \topr 0, \quad n,m \rightarrow \infty,
\end{equation}
and
\begin{equation}
\label{asnorm2}
\left( \frac{V}{n} + \frac{W}{m} \right)^{- \frac{1}{2}} \grad \lanm(0)
 \rightarrow_d \mathcal{N} (0,I), \quad n,m \rightarrow \infty.
\end{equation}
We start with \eqref{sup2}.
Using Taylor expansion the left hand side of (\ref{sup2}) can be bounded by
\begin{equation}\label{sup21}
\frac{M^2}{2}\!\! \left( \sup_{\th \in U^m_n} \left|\hess \lan(\th)-D(\th)\right| + \sup_{\th \in U_n^m}\left|D(\th)-D(0)\right|  +
 \sup_{\th \in U^m_n} \left|\hess \rem_n(\th) \right| \right)
\end{equation}
for $U^m_n = \{ \th : |\th| \leq \frac{M}{\sqrt{n+m}} \}.$
First two terms in (\ref{sup21}) tend to zero in probability by
assumptions 4 and 5. We prove that assumptions 6 and 7 imply convergence
to zero in probability of the third term
 in (\ref{sup21}). Calculating  the second
derivative of $\rem_n(\th) $ we get
\begin{eqnarray*}\nonumber
\hess \rem_n(\th)&=&\frac{1}{n} \sum_{i=1}^n \left[ \frac{  \hess C_m(X_i,\th)}{ C_m(X_i,\th)}-\frac{\grad C_m(X_i,\th)  \grad\t C_m(X_i,\th)}{C_m^2(X_i,\th)} \right. \\
&-& \left. \frac{\hess \Z(X_i,\th)}{\Z(X_i, \th)}+\frac{\grad \Z(X_i, \th)\grad\t \Z(X_i, \th)}{\Z^2(X_i, \th) }\right].
\end{eqnarray*}
Therefore
\begin{eqnarray}
\label{sup22}
&&\sup_{\th \in U}  \left|\hess \rem_n(\th) \right| \leq \sup_{\th \in U, x \in \mX}
\frac{  |\hess C_m(x,\th)| \, |C_m(x,\th) - C(x,\th)|}{ C_m(x,\th) C(x,\th)} \\
&+&
\sup_{\th \in U, x \in \mX} \frac{  | \hess C_m(x, \th) - \hess C(x, \th)|}{ C(x, \th)}
\nonumber \\
&+& \sup_{\th \in U, x \in \mX}
\frac{  |\grad C_m(x, \th)|^2 \, |C_m^2(x, \th) - C^2(x, \th)|}{ C_m^2(x,\th) C^2(x,\th) }
\nonumber \\
&+&
\sup_{\th \in U, x \in \mX} \frac{  | \grad C_m(x,\th)  \grad\t C_m(x, \th)-
\grad C(x,\th)\grad\t C(x, \th)|}{ C^2(x,\th)} . \nonumber
\end{eqnarray}
The convergence in assumptions 6(a) and 6(b) can be strengthened to be uniform over $\th \in U$
in the similar way as in the proof of Theorem \ref{without_cov}. Using these arguments and assumption 7
we obtain that for arbitrary $\eta >0$ and sufficiently large $m$ 
with probability at least $ 1 - \eta$
for each $x \in \mX, \theta \in U$
\[
\alpha/2\leq C_m(x,\th)\leq K+\alpha/2.
\]
Hence, every term on the right side of (\ref{sup22}) tends to zero in probability.

The last step is proving (\ref{asnorm2}).
 First, notice that
\begin{eqnarray}
\sqrt{n+m} \grad \lanm(0) &=& \sqrt{\frac{n+m}{n}} \: \sqrt{n} \grad \lan(0)
- \sqrt{\frac{n+m}{m}} \: \sqrt{m} \grad \renm(0) \nonumber \\ \label{asymp1}
&=& \left[ \sqrt{\frac{n+m}{n}} \: \sqrt{n} \grad \lan(0) - \sqrt{\frac{n+m}{m}} \frac{1}{ \sqrt{m}} \sum_{k=1}^m \bar{\phi} (Y^k)
\right]  \\ \label{asymp111}
&+& \sqrt{\frac{n+m}{m}} \left[ \frac{1}{ \sqrt{m}} \sum_{k=1}^m \bar{\phi} (Y^k)
-\sqrt{m} \grad \renm(0) \right],
\end{eqnarray}
where $\bar{\phi} (y) = \Ex_{X \sim g} \phi(y|X).$
By CLT, the expression  (\ref{asymp1}) tends in distribution to
$\mN(0, V/a + W/(1-a)), $ since the Monte Carlo sample is independent of the observation.
To show that the term (\ref{asymp111}) tends to zero in probability, we prove that
\begin{equation}\label{asymp12}
\sqrt{m} \grad \renm(0) - \frac{1}{n}\sum_{i=1}^n \frac{1}{ \sqrt{m}} \sum_{k=1}^m
\phi(Y^k| X_i)
\end{equation}
and
\begin{equation}\label{asymp13}
\frac{1}{n}\sum_{i=1}^n \frac{1}{ \sqrt{m}} \sum_{k=1}^m
\phi(Y^k| X_i) -\frac{1}{ \sqrt{m}} \sum_{k=1}^m \bar{\phi} (Y^k)
\end{equation}
tends to zero in probability. We start with (\ref{asymp12}) and calculate
\begin{equation}\nonumber
\grad\rem_n(0)
               =\frac{1}{n}\sum_{i=1}^{n }\frac{\dfrac{1}{m}\sum\limits_{k=1}^{m}
\phi(Y^k |X_i) \, \Z (X_i,0)}
                 {C_m(X_i,0)}\;.
\end{equation}
Therefore, by Cauchy-Schwarz inequality, expression (\ref{asymp12}) is
bounded by
\begin{equation}
\label{asymp121}
\sqrt{\frac{1}{n} \sum_{i=1}^n \frac{[C_m(X_i,0) - C(X_i,0)]^2}{C^2_m(X_i,0)}}\:
\sqrt{\frac{1}{n} \sum_{i=1}^n \left| \frac{1}{ \sqrt{m}} \sum_{k=1}^m
\phi(Y^k| X_i) \right|^2}.
\end{equation}
By assumptions 6(a) and 7 we again obtain that for arbitrary
$\varepsilon>0, \eta >0$ and sufficiently large $m$ with probability at least $1-\eta$
for every $x \in \mX$
\begin{equation}\nonumber
|\Z_m (x,0) -\Z(x,0)|\leq \varepsilon \quad {\rm and } \quad C_m(x,0) 
\geq  \alpha/2.
\end{equation}

Therefore, the term under the first square root in (\ref{asymp121}) tends in probability to zero, because with probability at least $1-\eta$
\begin{equation}\nonumber
\frac{1}{n} \sum_{i=1}^n \frac{[C_m(X_i,0) - C(X_i,0)]^2}{C^2_m(X_i,0)} \leq \sup_{x \in \mX} \frac{[C_m(x,0) - C(x,0)]^2}{C^2_m(x,0)} \leq \frac{4 \varepsilon^2}{\alpha ^2} \:.
\end{equation}
Using Markov's inequality and assumption 3  the second square root is bounded in probability, since
\begin{eqnarray*}
\Ex_{X_i \sim g,Y^k \sim h} \frac{1}{n} \sum_{i=1}^n \left| \frac{1}{ \sqrt{m}} \sum_{k=1}^m
\phi(Y^k| X_i) \right|^2 &=& \Ex_{X \sim g ,Y^k \sim h}  \left| \frac{1}{ \sqrt{m}} \sum_{k=1}^m
\phi(Y^k| X) \right|^2 \\
&=& \Ex_{X \sim g ,Y \sim h}  \left|
\phi(Y| X) \right|^2 = \tilde W < \infty,
\end{eqnarray*}
where we use the fact that $\Ex_{Y \sim h} \phi (Y|x)=0$ for fixed $x.$
Now consider (\ref{asymp13}). Change the order of summation and notice that
\begin{eqnarray*}
&&\Ex_{X_i \sim g ,Y^k \sim h} \left|\frac{1}{ \sqrt{m}} \sum_{k=1}^m  \left[\frac{1}{n}\sum_{i=1}^n
\phi(Y^k| X_i)   -  \bar{\phi}(Y^k)\right]\right|^2 \\
&=& \Ex_{X_i \sim g,Y \sim h} \left|\frac{1}{n}\sum_{i=1}^n
\phi(Y| X_i)  - \bar{\phi}(Y)\right|^2
=\frac{1}{n} \Ex_{X \sim g ,Y \sim h} \left|\phi(Y| X)  - \bar{\phi}(Y)\right|^2,
\end{eqnarray*}
so (\ref{asymp13}) tends to zero in $L^2$, hence, in probability.
\end{proof}

Finally, we discuss assumptions in Theorem \ref{with_cov}. 
Note that conditions 1-3  are similar to their analogs
in Theorem \ref{without_cov}. Therefore, we
briefly comment on the others. Consider the exponential family
$$
f(y|x, \th)= \exp( \th\t W(y,x))
$$
where $W(y,x) = \left( W_1(y,x), \ldots, W_p(y,x)\right),$ the set $ \mX$ is compact and  the function $W(y,x)$ is continuous with respect to the variable $x$. 
For simplicity we restrict attention to finite (but very large) space $\mathcal{Y}$, so that
$$
\Z(x,\th) = \sum_{y \in \mathcal{Y}} \exp( \th\t W(y,x)).
$$
The  autologistic model \cite{HuWu1998} that is very popular in spatial statistics belongs to this family. We can calculate that
\begin{eqnarray*}
\grad \Z(x,\th) &=& \sum_{y \in \mathcal{Y}} \exp( \th\t W(y,x)) W(y,x)\\
\hess \Z(x,\th) &=& \sum_{y \in \mathcal{Y}} \exp( \th\t W(y,x)) W(y,x) W^T(y,x)\\
\hess \log p(y|x,\th)&=& -\hess \log \Z (x,\th)= -\frac{\hess \Z (x,\th)}{\Z(x,\th)}
+ \frac{\grad \Z (x,\th) \grad ^T \Z (x,\th)}{\Z ^2(x,\th)}\:.
\end{eqnarray*}
Since the function $W(y,x)$ is continuous with respect to $x$ functions $\Z(x,\th),$
$\grad \Z(x,\th),\hess \Z(x,\th)$ are continuous with respect to both variables on the compact set $\mathcal{X} \times U,$
therefore assumption 7 is satisfied. Besides, the function $\hess \log p(y|x,\th)$ is also continuous that implies condition 4.  
The uniform convergence in assumption 5 and 6 follows 
from \cite[Theorem 16(a)]{Ferguson1996} or \cite[Example 19.8]{vaart1998} if we again use compactness of sets $\mathcal{X}, U $ and continuity of considered functions.

\end{document}